\documentclass[12pt, a4paper]{article}
\usepackage[utf8]{inputenc}
\usepackage{dcolumn,lscape}%,setspace}
\usepackage{amsmath,longtable,multicol,dcolumn,tabularx,graphicx,amssymb}
\usepackage{exscale,amsthm,multirow,rotating}
\usepackage{natbib}
\usepackage{bbm}
\usepackage{tikz,dsfont,float}
\usepackage[caption = false]{subfig}
\newcommand{\xvec}{\boldsymbol}
\newcommand{\xmat}{\mathbf}
\newcommand{\xset}{\mathds}

\newtheorem{theorem}{Theorem}
\newtheorem{example}{Example}

\newtheorem{corollary}{Corollary}

\usetikzlibrary{calc}
\usetikzlibrary{positioning}

\begin{document}

\def\spacingset#1{\renewcommand{\baselinestretch}%
{#1}\small\normalsize} \spacingset{1}

%%%%%%%%%%%%%%%%%%%%%%%%%%%%%%%%%%%%%%%%%%%%%%%%%%%%%%%%%%%%%%%%%%%%%%%%%%%%%%

\title{\bf Spatial and Spatiotemporal GARCH Models - A Unified Approach}
  \author{Philipp Otto\\
    \small{Leibniz University Hannover, Germany}\\
    Wolfgang Schmid\\
    \small{European University Viadrina, Frankfurt (Oder), Germany}}
  \maketitle

\begin{abstract}
In time-series analyses, particularly for finance, generalized autoregressive conditional heteroscedasticity (GARCH) models are widely applied statistical tools for modelling volatility clusters (i.e., periods of increased or decreased risk). In contrast, it has not been considered to be of critical importance until now to model spatial dependence in the conditional second moments. Only a few models have been proposed for modelling local clusters of increased risks. In this paper, we introduce novel spatial GARCH and exponential GARCH processes in a unified spatial and spatiotemporal GARCH-type model, which also covers all previously proposed spatial ARCH models as well as time-series GARCH models. For this common modelling framework, estimators are derived based on nonlinear least squares and on the maximum-likelihood approach. In addition to the theoretical contributions of this paper, we suggest a model selection strategy that is verified by a series of Monte Carlo simulation studies. Eventually, the use of the unified model is demonstrated by an empirical example that focuses on real estate prices from 1995 to 2014 across the ZIP-Code areas of Berlin. A spatial autoregressive model is applied to the data to illustrate how locally varying model uncertainties can be captured by the spatial GARCH-type models.
\end{abstract}

\noindent%
\vfill

% \newpage
\spacingset{1.45} % DON'T change the spacing!

\section{Introduction}\label{sec:introduction}

Recent literature have dealt with the extension of generalized autoregressive conditional heteroscedasticity (GARCH) models to spatial and spatiotemporal processes (e.g. \citealt{Otto16_arxiv,Otto18_spARCH,Otto19_statpapers,Sato17,Sato18b,Sato18a}). Whereas the classical ARCH model is defined as a process over time, these random processes have a multidimensional support. Thus, they allow for spatially dependent second-order moments, while the local means are uncorrelated and constant in space (see \citealt{Otto19_statpapers}). \cite{Sato17,Sato18a} introduced a random process incorporating elements of GARCH and exponential GARCH (E-GARCH) processes, which is, however, neither a GARCH nor an E-GARCH process. Moreover, \citealt{Otto16_arxiv} only focused on spatial ARCH processes without considering the influences from the realized, conditional variance at neighboring locations. Direct extensions of GARCH and E-GARCH processes to spatial settings do not exist among current research.

In the paper, we have several major contributions. First, we introduce completely novel spatial GARCH and E-GARCH processes. Along with the E-GARCH approach, a logarithmic spatial GARCH approach is discussed as special case. Secondly, all of these new spatial GARCH models are combined in a unified framework along with previously proposed models and time-series GARCH models. This unified spatial GARCH process is a completely new class of models in spatial econometrics, which allow for a common estimation strategies, unified model comparisons, and model diagnostics. In addition, all models are computationally implemented in one library, the R-package \texttt{spGARCH}. Due to the flexible structure of this approach, other GARCH-type models, like threshold or multivariate GARCH models, can easily be constructed.

From a practical perspective, this unified spatial GARCH model can be used to model spill-over effects in the conditional variances. That means that an increasing variance in a certain region of the considered space would lead to an increase or decrease in the adjacent regions, depending on the direction (sign) of the spatial dependence. Local climate risks, such as fluctuations in the temperature and precipitation, or financial risks in spatially constrained markets, such as the real estate or labor, could be modeled using this approach. Furthermore, spatial GARCH-type models can be used as error models for any linear or nonlinear spatial regression model to account for local model uncertainties (i.e., areas in which the considered models perform worse than in others). Such model uncertainties can be considered to be a kind of local risk.

The remainder of this paper is structured as follows. In the next section, we introduce the unified approach and discuss several examples nested within this approach, including novel spatial GARCH and E-GARCH processes. Following from there, different estimation procedures are briefly described. These theoretical sections are followed up with a discussion of the insights gained from simulation studies. The paper then supplies a real-world example, namely the real estate prices in the German capital city of Berlin. In Section \ref{sec:conclusion}, we stress some important extensions for future research before concluding the paper.

\section{Spatial and Spatiotemporal GARCH-Type Models}\label{sec:models}

Let $\left\{Y(\xvec{s}) \in \xset{R}: \xvec{s} \in D_{\xvec{s}} \right\}$ be a univariate stochastic process, where $D_{\xvec{s}}$ represents a set of possible locations in a $q$-dimensional space. Thus, spatial and spatiotemporal models are both covered by this approach. With regards to spatiotemporal processes, the temporal dimension can be easily considered as one of the $q$ dimensions. In addition, time-series GARCH models are included for $q = 1$.

Let $\xvec{s}_1, \ldots, \xvec{s}_n$ denote all locations, and let $\xvec{Y}$ stand for the vector of observations \linebreak $\left(Y\left(\xvec{s}_i\right)\right)_{i = 1, \ldots, n}$. The commonly applied spatial autoregressive (SAR) model implies that the conditional variance $Var(Y(\xvec{s}_i) | Y(\xvec{s}_j), j \neq i)$ is constant (cf. \citealt{Cressie93,Cressie11}) and does not depend on the observations of neighboring locations. This approach is extended by assuming the changes in the volatility can spill over to neighboring regions and that conditional variances can vary over space, resulting in clusters of high and low variance.
As in time-series ARCH models developed by \cite{Engle82}, the vector of observations is given by the nonlinear relationship
\begin{equation}
\xvec{Y} = \text{diag}(\xvec{h})^{1/2} \xvec{\varepsilon} \, \label{eq:initial}%%%
\end{equation}
where $\xvec{h} = (h(\xvec{s}_1), \ldots, h(\xvec{s}_n))'$ and $\xvec{\varepsilon} = (\varepsilon(\xvec{s}_1), \ldots, \varepsilon(\xvec{s}_n))'$ is a noise component, which is later specified in more detail.

\subsection{A Unified Approach}

For the unified approach, we assume that a known function $f$ exists, which relates $\xvec{h}$ to a vector $\xvec{F} = (f(h(\xvec{s}_1)), \ldots, f(h(\xvec{s}_n)))^\prime$. By choosing $f$ and a suitable model of $\xvec{F}$, different spatial GARCH-type processes can be defined. For instance, they could have additive or multiplicative dynamics, or the spill-over effects in the conditional variances could be global or locally constrained to direct neighboring observations. Over the following two sections of the paper, we suggest two alternative representations of the model of $\xvec{F}$. These are generally equivalent. We discuss the link between both approaches in more detail in Section \ref{sec:g_and_gamma}.

After introducing the spatial heteroscedasticity models $\xvec{F}$, we discuss examples of spatial GARCH-type models nested in the approach. Some of these models have already been discussed in recent literature (e.g., \citealt{Otto18_spARCH,Sato18b}; see also Examples \ref{example:spARCH}, \ref{example:spHARCH}). However, we also introduce completely novel spatial models, namely the spatial GARCH (analogous to \citealt{Bollerslev86}; see also Example \ref{example:spGARCH}), the exponential spatial GARCH (analogous to \citealt{Nelson91}; see also Example \ref{example:spEGARCH}), and the spatial log-GARCH model (analogous to \citealt{Pantula86,Geweke86,Milhoj87}; see also Example \ref{example:spLGARCH}).

\subsubsection{Spatial and Spatiotemporal GARCH Models of Type I}

Here we introduce a general approach covering some important spatial and spatiotemporal GARCH-type models, namely the spatial ARCH model of \cite{Otto16_arxiv,Otto18_spARCH} and the hybrid model of \cite{Sato18b}. We refer to this class of spatial and spatiotemporal GARCH models as Type I. The vector $\xvec{F}$ is chosen as
\begin{equation}\label{eq:unified}
\xvec{F} = \xvec{\alpha} + \xmat{W}_1 \xvec{\gamma}(\xvec{Y}^{(2)}) + \xmat{W}_2 \xvec{F}
\end{equation}
with a measurable function $\xvec{\gamma}(\xvec{x}) = (\gamma_1(\xvec{x}), \ldots, \gamma_n(\xvec{x}))^\prime$ and \linebreak
$\xvec{Y}^{(2)} = (Y(\xvec{s}_1)^2, \ldots , Y(\xvec{s}_n)^2)^\prime$. The weighting matrices $\xmat{W}_1 = (w_{1,ij})_{i,j = 1, \ldots, n}$ and \linebreak $\xmat{W}_2 = (w_{2,ij})_{i,j = 1, \ldots, n}$ are assumed to be non-negative with zeros on the diagonal (i.e., $w_{v,ij} \ge 0$ and $w_{v,ii}=0$ for all $i,j = 1, \ldots, n$ and $v = 1, 2$). Moreover, let $\xvec{\alpha} = (\alpha_i)_{i = 1, \ldots, n}$ be a positive vector.

First, we discuss under what conditions the process is well-defined. In the following we make use of the notation $\xmat{E} = \mbox{diag}(\varepsilon(\xvec{s}_1)^2, \ldots, \varepsilon(\xvec{s}_n)^2)$. The proofs of all results are provided in the Appendix.

\begin{theorem}\label{solution}
Let $\xvec{z} = ( z_i )_{i=1,...,n}$ and suppose that for $z_i \ge 0$, $i=1,...,n$
\begin{equation}\label{eq:T}
\xvec{T}(\xvec{z}) = \xvec{\alpha} + \xmat{W}_1 \xvec{\gamma}(\xmat{E} \xvec{z}) - (\xmat{I} - \xmat{W}_2) ( f(z_i) )_{i=1,...,n} + \xvec{z}
\end{equation}
is a contraction mapping, i.e. that there exists $L \in [0,1)$ with $|| \xvec{T}(\xvec{z}_1) - \xvec{T}(\xvec{z}_2)|| \le L ||\xvec{z}_1 - \xvec{z}_2||$, then the equations (1) and (2) have exactly one real-valued solution $\xvec{Y}$.
\end{theorem}

Note that the condition (\ref{eq:T}) is fulfilled if, for example, $\xvec{\gamma}$ satisfies a Lipschitz condition with constant $L_1$, $(f(z_i))_{i=1,...,n}$ satisfies a Lipschitz condition with a constant $L_2$ and $L := L_1 ||\xmat{W}_1 \xmat{E}|| + L_2 ||\xmat{I} - \xmat{W}_2|| < 1$. In addition, the fixed-point theorem of Banach implies that the sequence $\xvec{z}_m = \xvec{T}(\xvec{z}_{m-1})$, $m \ge 1$ converges to $\xvec{h}$ for given $\xvec{\alpha}$, $\xmat{E}$, $\xmat{W}_1$, and $\xmat{W}_2$. Consequently, this result resents a possibility to simulate such a process.
Below, we discuss some important properties of this process including the following condition for stationarity.

\begin{corollary}\label{cor:stationarity2}
Suppose that the assumptions of Theorem \ref{solution} are fulfilled. If $(\varepsilon(\xvec{s}_1),\ldots,\varepsilon(\xvec{s}_n))^\prime$ is strictly stationary, then $(Y(\xvec{s}_1), \ldots, Y(\xvec{s}_n))^\prime$ is strictly stationary as well.
\end{corollary}

Moreover, the observations $Y(\xvec{s})$ are uncorrelated with a mean of zero, as we will show in the following theorem. Thus, spatial GARCH models are suitable error models for use with other linear or nonlinear spatial regression models, such as spatial autoregressive models, without affecting the mean equation. In this way, locally varying model uncertainties can be captured.

\begin{theorem}\label{th:moments1}
Let $i \in \{1,\ldots,n\}$. Suppose that the assumptions of Theorem \ref{solution} are satisfied and that $\xvec{\varepsilon}$ is sign-symmetric, i.e.,
\[ \xvec{\varepsilon} \stackrel{d}{=} ((-1)^{v_1} \varepsilon(\xvec{s}_1),\ldots,(-1)^{v_n} \varepsilon(\xvec{s}_n)) \qquad \mbox{for all} \qquad v_1,\ldots,v_n \in \{0,1\} . \]
\begin{itemize}
\item[a)] Then $Y(\xvec{s}_i)$ is a symmetric random variable. All odd moments and all conditional odd moments of $Y(\xvec{s}_i)$ are zero, provided that they exist.
\item[b)] It holds that $\mbox{Cov}(Y(\xvec{s}_i), Y(\xvec{s}_j)) = 0$ for $i \neq j$ if the second moment exists.
\end{itemize}
\end{theorem}

In the spatial setting, however, the conditional variance $Var(Y(\xvec{s}_i) | Y(\xvec{s}_j), j \neq i)$ is not exactly equal to $h(\xvec{s}_i)$  (see \citealt{Otto19_statpapers}). Nevertheless, the interpretation of $\xvec{h}$ is similar to the conditional variance. In locations $\xvec{s}$, where $h(\xvec{s})$ is large, the conditional variance is also large and vice versa (see \citealt{Otto19_statpapers}, Fig. 1). That means that the local risk or level of uncertainty of this particular region is high compared to its neighbors. Such regions could be identified via $\xvec{h}$; this could be of interest in terms of the valuation of real estate or other immovable assets, since it provides insights ínto an individual location's risk. In addition to this, the spatial GARCH coefficients measure potential risk spillovers from neighboring locations. It is worth noting that in the case of directional spatial processes, $\xvec{h}$ is equal to the conditional variances. Thus, it can be interpreted in the same way as with time-series GARCH models (see \citealt{Otto19_statpapers}).

This Type-I approach allows for a large range of GARCH-type models. Depending on the definition of $f$ and $\xvec{\gamma}$, the resulting spatial GARCH-type models have different stochastic properties. We discuss some important special cases below, starting with the spatial ARCH model (\citealt{Otto16_arxiv,Otto18_spARCH,Otto19_statpapers}), which is a direct extension of the ARCH process of \cite{Engle82} to spatial and spatiotemporal processes. It was originally introduced by \cite{Otto16_arxiv}. For more details on its stochastic properties, we refer to \citealt{Otto19_statpapers}.

\begin{example}[Spatial ARCH process of \citealt{Otto16_arxiv}]\label{example:spARCH}
Choosing $f(x) = x$, $\gamma_i(\xvec{x}) = x_i$ for $i = 1, \ldots, n$, and $\xmat{W}_2 = \xmat{0}$ the spatial ARCH (spARCH) process is obtained. It is given by
\begin{equation*}
  Y(\xvec{s}_i) = \sqrt{h(\xvec{s}_i)} \varepsilon(\xvec{s}_i), \quad i=1,...,n
\end{equation*}
with
\begin{equation*}
\xvec{h} = \xvec{\alpha} + \xmat{W}_1 \xvec{Y}^{(2)}\, .
\end{equation*}
\end{example}

Indeed, the spatial ARCH process can be easily extended to a spatial GARCH process by considering the realized values of $h(\cdot)$ in adjacent locations.

\begin{example}[Spatial GARCH process]\label{example:spGARCH}
Taking $f(x) = x$ and  $\gamma_i(\xvec{x}) = x_i$ for $i=1,...,n$ a spatial  GARCH (spGARCH) process is obtained. That is,
\begin{equation*}
  Y(\xvec{s}_i) = \sqrt{h(\xvec{s}_i)} \varepsilon(\xvec{s}_i), \quad i=1,...,n
\end{equation*}
with
\begin{equation*}
\xvec{h} = \xvec{\alpha} + \xmat{W}_1 \xvec{Y^{(2)}} + \xmat{W}_2 \xvec{h} \, .
\end{equation*}
\end{example}

Since $\xvec{Y}^{(2)} = \xmat{E} \xvec{h}$, the quantity $\xvec{h}$ can be specified as
\begin{equation}\label{epsGARCH}
\xvec{h} = ( \xmat{I} - \xmat{W}_1 \xmat{E} - \xmat{W}_2)^{-1} \xvec{\alpha} \, ,
\end{equation}
if the inverse exists.  For this simple example, there is a unique solution if $||\xmat{I} - \xmat{W}_1 \xmat{E} - \xmat{W}_2|| < 1$, as it is already expressed in Theorem \ref{solution}.

Contrary to this approach, \cite{Sato17,Sato18b,Sato18a} have considered a slightly different choice of $\xvec{h}$, and have used the log-transformation to avoid any non-negativity problems of $\xvec{h}$. Thus, their model combines the GARCH and the E-GARCH attempts. For that reason, we will denote it as the hybrid model. Let $\xvec{h}_L = ( \log(h(\xvec{s}_i)) )_{i=1,...,n}$ and $\xvec{Y}^{(2)}_L = ( \log(Y(\xvec{s}_j)^2 ) )_{i=1,...,n}$.

\begin{example}[Hybrid spatial GARCH process of \citealt{Sato17}]\label{example:spHARCH}
Choosing $f(x) = \log(x)$ and $\gamma_i(\xvec{x}) = \log(x_i)$ the hybrid spatial GARCH (H-spGARCH) process is obtained, i.e.,
\begin{equation*}
  Y(\xvec{s}_i) = \sqrt{h(\xvec{s}_i)} \varepsilon(\xvec{s}_i), \quad i=1,...,n
\end{equation*}
with
\begin{equation*}
\xvec{h}_L = \xvec{\alpha} + \xmat{W}_1  \xvec{Y}^{(2)}_L + \xmat{W}_2 \xvec{h}_L \, .
\end{equation*}
\end{example}

The process has a unique solution if $||\xmat{I} - \xmat{W}_1 - \xmat{W}_2|| < 1$. This is an immediate consequence of Theorem \ref{solution}. It is obtained by setting $\gamma(\xvec{x}) = (\log(x_i) )_{i=1,\ldots,n}$, $f(x)=\log(x)$, and considering the right side of \eqref{eq:T} to be a function of $\log(z)$. We see that the condition on the existence of a solution is much simpler than for the spGARCH process, since it only depends on the weight matrices and not on the random matrix $\xmat{E}$. This simplification is due to the fact that we have an additive decomposition of the function $\xvec{\gamma}$, i.e., $\xvec{\gamma}(\xmat{E} \xvec{z}) = \xvec{\phi}_1(\xmat{E}) + \xvec{\phi}_2(\xvec{z})$ with certain functions $\xvec{\phi}_1$ and $\xvec{\phi}_2$. This functional equation is solved by  the logarithm function. However, the behavior of the H-spGARCH is different to that of the spGARCH. Thus, the one or the other could be preferable for empirical applications.

\subsubsection{Spatial and Spatiotemporal GARCH Models of Type II}

In this section, we extend the exponential and logarithmic GARCH models to apply these to spatial and spatiotemporal data (see Examples \ref{example:spEGARCH} and \ref{example:spLGARCH}). Like for the GARCH Type I approach, these models represent special cases of a new more general attempt. We call this class of spatial and spatiotemporal models GARCH models, Type II GARCH models. Here, the function $\xvec{\gamma}(\xvec{Y}^{(2)})$ in \eqref{eq:unified} is replaced by  the quantity $\xvec{g}(\xvec{\varepsilon}) = (g_1(\xvec{\varepsilon}), \ldots, g_n(\xvec{\varepsilon}))^\prime$. This leads to
\begin{equation}\label{eq:unified2}
\xvec{F} = \xvec{\alpha} + \xmat{W}_1 \xvec{g}(\xvec{\varepsilon}) + \xmat{W}_2 \xvec{F} \, .
\end{equation}

The relationship between both approaches will be discussed in detail in the next section. In general, the the Type I and Type II approaches have similar properties; for example, zero odd moments and uncorrelated observations. Below, we derive the results of the Type II approach in the same structure as above.

\begin{theorem}\label{th:existence}
Let $\xvec{F} = (F_i)_{i = 1, \ldots, n}$. Suppose that $rk(\xmat{I} - \xmat{W}_2) = n$ and that $f:(0,\infty) \rightarrow \xset{R}$ has an inverse function, then Equations \eqref{eq:initial} and \eqref{eq:unified} have exactly one real-valued solution $\xvec{Y}$ given by $Y(\xvec{s}_i) = \sqrt{f^{-1}(F_i)} \; \varepsilon(\xvec{s}_i)$ for $i = 1, \ldots, n$.
\end{theorem}

To ensure that an inverse function $f^{-1}$ exists, it is sufficient to assume that $f$ is a continuous and strictly increasing function on $(0, \infty)$. Spatial GARCH processes of this type are strictly stationary, if the errors $\xvec{\varepsilon}$ are strictly stationary.

\begin{corollary}\label{cor:stationarity}
Suppose that the assumptions of Theorem \ref{th:existence} are fulfilled and that $f^{-1}$ is a measurable function. If $(\varepsilon(\xvec{s}_1),\ldots,\varepsilon(\xvec{s}_n))^\prime$ is strictly stationary, then $(Y(\xvec{s}_1), \ldots, Y(\xvec{s}_n))^\prime$ is strictly stationary as well.
\end{corollary}

Before discussing some important special cases, we derive the conditional and unconditional moments of the process $\xvec{Y}$. Let $c_i$ denote the $i$-th component of the vector $(\xmat{I} - \xmat{W}_2)^{-1} \xvec{\alpha}$, and $\xvec{d}_i^\prime$ indicates the $i$-th row of matrix $(\xmat{I} - \xmat{W}_2)^{-1}  \xmat{W}_1$. Then, $f(h(\xvec{s}_i)) = c_i + \xvec{d}_i^\prime \xvec{g}(\xvec{\varepsilon})$.

\begin{theorem}\label{th:moments}
Let $i \in \{1,\ldots,n\}$. Suppose that the assumptions of Theorem \ref{th:existence} are satisfied.
\begin{itemize}
\item[a)] Let $k \in I\!\!N$. Suppose that $f^{-1}$ is a convex function. If the $2k$-th moments of $\xvec{\varepsilon}$ exist and $E\left( \left( f^{-1}(2 \xvec{d}_i^\prime g(\xvec{\varepsilon}) ) \right)^{2k} \right)  < \infty$, then the $k$-th moment of $Y(\xvec{s}_i)$ exists.
\item[b)] If $\xvec{\varepsilon}$ is sign-symmetric, then $Y(\xvec{s}_i)$ is a symmetric random variable. All odd moments and all conditional odd moments of $Y(\xvec{s}_i)$ are zero, provided that they exist.
\item[c)] If the conditions of Theorem \ref{th:existence} are fulfilled and if $\xvec{\varepsilon}$ is sign-symmetric then \linebreak $\mbox{Cov}(Y(\xvec{s}_i), Y(\xvec{s}_j)) = 0$ for $i \neq j$.
\end{itemize}
\end{theorem}

At first, we extend the E-GARCH model to spatial and spatiotemporal processes. In a time-series analysis, for instance, the aim of the E-GARCH models is to describe asymmetries in a return process. These are caused by the leverage effect whereby an investor reacts differently to negative and positive news (see, e.g., \citealt{Day92,Engle93,Heynen94}). Such a behavior may also arise in spatial settings.

\begin{example}[Exponential spatial GARCH process]\label{example:spEGARCH}
Taking $f(x) = \log(x)$ and $g_i(\xvec{x}) = \Theta x_i + \zeta (|x_i| - E(|x_i|) ) = g(x_i)$ for all $i = 1, \ldots, n$ and $\Theta > 0$, $\zeta \geq 0$ in \eqref{eq:unified2} the exponential spatial GARCH process,  briefly E-spGARCH, is obtained. It holds that
\begin{equation*}
  Y(\xvec{s}_i) = \sqrt{h(\xvec{s}_i)} \varepsilon(\xvec{s}_i) , \quad i=1,...,n
\end{equation*}
with
\begin{equation}\label{eq:spEGARCH2}
\xvec{h}_L = \xvec{\alpha} + \xmat{W}_1  \xvec{G}(\xvec{\varepsilon}) + \xmat{W}_2 \xvec{h}_L
\end{equation}
where $\xvec{G}(\xvec{\varepsilon}) = ( g(\varepsilon(\xvec{s}_i)) )_{i=1,...,n}$.
\end{example}

Taking the exponentials on both sides of (\ref{eq:spEGARCH2}), we get  that
\begin{equation}\label{eq:spEGARCH}
h(\xvec{s}_i) = \exp(\alpha_i) \; \left( \prod_{j=1}^n h(\xvec{s}_j)^{w_{2,ij}} \right) \; \left( \prod_{j=1}^n \exp(w_{1,ij} g(\varepsilon(\xvec{s}_j))) \right) , \quad i = 1, \ldots, n .
\end{equation}
This shows that the process possesses multiplicative dynamics and, in the case of an spGARCH process, additive dynamics (see also \citealt{Francq11} regarding time series GARCH models). Furthermore, a logarithmic spatial GARCH (log-spGARCH) process can be defined that is analogous of the log-GARCH process of  \cite{Pantula86,Geweke86,Milhoj87}. Note that the hybrid model collapses to a log-spGARCH model, if $\xmat{W}_1$ and $\xmat{W}_2$ are identical.

\begin{example}[Spatial logarithmic GARCH process]\label{example:spLGARCH}
Choosing  $f(x) = \log(x)$ and $g_i(\xvec{x}) = \log(|x_i|^b)$ for all $i = 1, \ldots, n$ and $b > 0$ in \eqref{eq:unified2} the spatial
logarithmic spatial GARCH process, briefly log-spGARCH. That is,
\begin{equation*}
  Y(\xvec{s}_i) = \sqrt{h(\xvec{s}_i)} \varepsilon(\xvec{s}_i), \quad i=1,..., n
\end{equation*}
with
\begin{equation}\label{eq:spLGARCH}
\xvec{h}_L = \xvec{\alpha} + \xmat{W}_1 \xvec{G}_L + \xmat{W}_2 \xvec{h}_L
\end{equation}
where $\xvec{G}_L = (\log(|\varepsilon(\xvec{s}_i)|^b) )_{i=1,...,n}$.
\end{example}

\subsubsection{Relationship between GARCH Models of Type I and Type II}\label{sec:g_and_gamma}

The approaches presented above have a similar structure. In the Type I approach, $\xvec{F}$ is considered to be a function of $\xvec{Y}^{(2)}$, while in the Type II attempt it is  a function of $\xvec{\varepsilon}$. Certainly, the question arises as to whether there is a relationship between both attempts.

First, suppose that a GARCH Type I process is given. Then
\[ \xvec{F} = \xvec{\alpha} + \xmat{W}_1 \xvec{\gamma}(\xmat{E} \xvec{h}) + \xmat{W}_2 \xvec{F} . \]
In Theorem \ref{solution}, a sufficient condition for the existence of a solution $\xvec{h}$ to this equation is given. If it can be solved, then $\xvec{h}$ is a function of $\xmat{E}$, $\xvec{\alpha}$, $\xmat{W}_1$, and $\xmat{W}_2$. Thus,
\[ \xvec{g}(\xvec{\varepsilon}) = \xvec{\gamma}(\xmat{E} \; \xvec{h}(\xmat{E}, \xvec{\alpha}, \xmat{W}_1, \xmat{W}_2))  \]
and the GARCH Type I model can be written as a GARCH Type II model. However, the structure of the function $\xvec{g}$ gets more complicated depending on the weight matrices.
For instance, for the spGARCH, we get that
\[ \xvec{\gamma}(\xvec{Y}^{(2)}) = \xvec{Y}^{(2)} = \xmat{E} \xvec{h} =
 \xmat{E} ( \xmat{I} - \xmat{W}_1 \xmat{E} - \xmat{W}_2)^{-1} \xvec{\alpha} = \xvec{g}(\xvec{\varepsilon}) .  \]
This shows that $\xvec{g}$ is a function depending on $\xmat{W}_1$, $\xmat{W}_2$, and $\xvec{\alpha}$.

For the H-spGARCH process, it follows with $\xvec{Y}^{(2)}_L = (\log(Y(\xvec{s}))^2,\ldots, \log(Y(\xvec{s}_n)^2))^\prime$ and
$\xvec{\varepsilon}^{(2)}_L = (\log(\varepsilon(\xvec{s}_1)^2,\ldots,\log(\varepsilon(\xvec{s}_n)^2)))^\prime$ that
$\xvec{\gamma}(\xvec{Y}^{(2)}) = \xvec{Y}_L^{(2)} = \xvec{h} + \xvec{\varepsilon}_L^{(2)}$.
Thus,\vspace*{-.7cm}

\begin{equation}\label{epsHARCH}
\xvec{h} = \xvec{\alpha} + (\xmat{W}_1 + \xmat{W}_2) \xvec{h} + \xmat{W}_1 \xvec{\varepsilon}^{(2)}_L = (\xmat{I} - \xmat{W}_1 - \xmat{W}_2)^{-1} (\xvec{\alpha} + \xmat{W}_1 \xvec{\varepsilon}^{(2)}_L )
\end{equation}
and, as a consequence,
\begin{equation*}
\xvec{\gamma}(\xvec{Y}^{(2)}) =  (\xmat{I} - \xmat{W}_1 - \xmat{W}_2)^{-1} \xvec{\alpha} + (\xmat{I} + (\xmat{I} - \xmat{W}_1 - \xmat{W}_2)^{-1} \xmat{W}_1 ) \xvec{\varepsilon}^{(2)}_L  = \xvec{g}(\xvec{\varepsilon}).
\end{equation*}

Secondly, let us assume that a GARCH Type II model is given and we want to transform it into a GARCH Type I model. Then
\begin{equation}\label{eq:transform}
\xvec{F} = \xvec{\alpha} + \xmat{W}_1 \xvec{g}( (Y(\xvec{s}_i)/\sqrt{h(\xvec{s}_i)} )_{i=1,...,n} ) + \xmat{W}_2 \xvec{F}  .
\end{equation}
This equation has to be solved with respect to $\xvec{h}$ so that $\xvec{h}$ becomes a function of $\xvec{Y}$. If there is a solution then
\[   g(\xvec{\varepsilon}) = \xvec{g}( (Y(\xvec{s}_i)/\sqrt{h(\xvec{s}_i)} )_{i=1,...,n} )  \]
and we have the desired relationship.  General statements about the existence of a solution $\xvec{h}$ in \eqref{eq:transform} are difficult to make, but of course, the theorem of Banach could once again be applied.

For the examples considered in the last section, the situation turns out to be simpler. Recall that $h(\xvec{s}_i) = Y(\xvec{s}_i)^2/\varepsilon(\xvec{s}_i)^2$. Thus, $\log(h(\xvec{s}_i)) = \log(Y(\xvec{s}_i)^2) - \log(\varepsilon(\xvec{s}_i)^2)$ (i.e., $\xvec{h}_L = \xvec{Y}^{(2)}_L - \xvec{\varepsilon}^{(2)}_L$). Inserting this into \eqref{eq:spEGARCH2} leads to
\begin{equation}\label{spEGARCH3}
(\xmat{I} - \xmat{W}_2) \xvec{Y}^{(2)}_L - \xvec{\alpha} = \xmat{W}_1 \xvec{G} + (\xmat{I} - \xmat{W}_2) \xvec{\varepsilon}^{(2)}_L,
\end{equation}
where $\xvec{G} = (g(\varepsilon(\xvec{s}_1)), \ldots, g(\varepsilon(\xvec{s}_n)))'$.
Thus, the right side of \eqref{spEGARCH3} is a function of $\xvec{\varepsilon}$, say $l(\xvec{\varepsilon})$. Suppose that $l$ is invertible. Then, we obtain
\begin{equation}\label{eq:inverse_f}%\label{inverse}
 \xvec{\varepsilon} = l^{-1}((\xmat{I} - \xmat{W}_2) \xvec{Y}^{(2)}_L - \xvec{\alpha})
\end{equation}
and $\gamma_i(\xvec{Y}) = g(\varepsilon(\xvec{s}_i))$.

Finally, for the log-spGARCH, an explicit solution of \eqref{eq:spLGARCH} exists because
\begin{equation*}
 \xvec{h}_L = (\xmat{I} + 0.5b \xmat{W}_1 - \xmat{W}_2)^{-1} (\xvec{\alpha} + b \xmat{W}_1 {\xvec{Y}}_L) \, ,
\end{equation*}
where ${\xvec{Y}}_L = (\log(|Y(\xvec{s}_1)|),\ldots, \log(|Y(\xvec{s}_n)|))^\prime$.

In summary, GARCH Type I models can be more easily transformed into a GARCH Type II models than vice versa. For the examples mentioned above, we derived the corresponding transformations. Finally, we give an overview on selected nested models in Figure \ref{fig:overview}. Some of these show additive dynamics (red, dashed box), while the other models show multiplicative dynamics (blue, dashed box). Of course, this overview is not exhaustive and there are many other spatial GARCH-type models that could be nested in this unified approach, such as threshold or multivariate GARCH models.

%\begin{landscape}
%\begin{figure}
\begin{sidewaysfigure}
  \centering
  \begin{tikzpicture}[
      grow                    = right,
      sibling distance        = 3cm,
      level distance          = 9cm,
      node distance           = 2.5cm,
      block/.style={
          rectangle,
          thick,
          align=center,
          minimum height=2em,
          minimum width=7.5cm
    }]
    \node (y) [block, draw] {$\xvec{Y} = \text{diag}(\xvec{h})^{1/2} \xvec{\varepsilon}$  \\ with $\xvec{F} = f(\xvec{h})$ and};
    \node (general_gamma) [block, draw, below of = y] {Type I: $\quad$ $\xvec{F} = \xvec{\alpha} + \xmat{W}_1 \xvec{\gamma}(\xvec{Y}) + \xmat{W}_2 \xvec{F}$, or}
        child { node [draw, block] (H-spGARCH) {Hybrid spatial GARCH model (Ex \ref{example:spHARCH}) \\ $\gamma_i(\xvec{x}) = \log(x_i)$}
                % edge from parent node [below] {$f(x) = \log(x)$}
                edge from parent[draw=none]
        }
        child { node [draw, block] (spGARCH) {\textbf{*} Spatial GARCH model (Ex \ref{example:spGARCH}) \\ $\gamma_i(\xvec{x}) = x_i$}
                % edge from parent node [above] {$f(x) = x$}
                edge from parent[draw=none]
        }
    ;
    \node (general_g) [block, draw, below = 4cm of general_gamma] {Type II: $\quad$ $\xvec{F} = \xvec{\alpha} + \xmat{W}_1 \xvec{g}(\xvec{\varepsilon}) + \xmat{W}_2 \xvec{F}$}
        child { node [draw, block] (E-spGARCH) {\textbf{*} Exponential spatial GARCH model (Ex \ref{example:spEGARCH})\\ $g_i(\xvec{x}) = \Theta x_i + \zeta (|x_i| - E(|x_i|) )$}
                % edge from parent node [above] {$f(x) = \log(x)$}
                edge from parent[draw=none]
        }
        child { node [draw, block] (log-spGARCH) {\textbf{*} Logarithmic spatial GARCH model (Ex \ref{example:spLGARCH}) \\ $g_i(\xvec{x}) = \log(|x_i|^b)$}
                % edge from parent node [below] {$f(x) = \log(x)$}
                edge from parent[draw=none]
        }
    ;
    \node [draw, block, below = 0.7cm of spGARCH, anchor = west] (spARCH) {Spatial ARCH model (Ex \ref{example:spARCH}) \\ $\xmat{W}_2 = \xmat{0}$};
    \node [draw, block, below = 0.7cm of log-spGARCH, anchor = west] (log-spARCH) {Logarithmic spatial ARCH model \\ $\xmat{W}_2 = \xmat{0}$};

    \draw[red,thick,dashed] ($(spGARCH)+(-3.9,0.8)$)  rectangle ($(spARCH)+(3.9,-0.8)$);
    \draw[blue,thick,dashed] ($(H-spGARCH)+(-3.9,0.8)$)  rectangle ($(log-spARCH)+(3.9,-2.5)$);
    \draw[red]               (general_gamma.east) -- (spGARCH.west) node[midway,fill=white] {\scriptsize{$f(x) = x$}};
    \draw[blue]              (general_gamma.east) -- (H-spGARCH.west) node[midway,fill=white] {\scriptsize{$f(x) = \log(x)$}}
                             (general_g.east) -- (E-spGARCH.west) node[midway,fill=white] {\scriptsize{$f(x) = \log(x)$}}
                             (general_g.east) -- (log-spGARCH.west) node[midway,fill=white] {\scriptsize{$f(x) = \log(x)$}};

    \draw[black, ->] ($(general_gamma.south)+(0.1,0)$) -- ($(general_g.north)+(0.1,0)$) node[midway,right] {\scriptsize{$\xvec{\gamma}(\xvec{Y}^{(2)}) = {\xvec{g}} (\xvec{\varepsilon}; \xvec{\alpha}, \xmat{W}_1, \xmat{W}_2, \ldots)$}} ;
    \draw[black, ->] (general_g.north) -- (general_gamma.south) node[midway,left] {\scriptsize{$\xvec{g}(\xvec{\varepsilon}) = {\xvec{\gamma}} (\xvec{Y}^{(2)}; \xvec{\alpha}, \xmat{W}_1, \xmat{W}_2, \ldots)$}} ;
\end{tikzpicture}
  \caption{Overview of selected models that are nested in the unified model, where all models in the red box show additive dynamics and the remaining models, framed by a dashed blue line, have multiplicative dynamics. Novel models proposed in this paper are indicated using an asterisk.}\label{fig:overview}
\end{sidewaysfigure}
%\end{figure}
%\end{landscape}

\section{Statistical Inference}

In the following section, we firstly discuss the choice of the weight matrices in more detail. In a general setting, $\xmat{W}_1$ and $\xmat{W}_2$ have $n(n-1)$ free parameters, while only $n$ values are observed. In spatial econometrics, these matrices are therefore usually replaced with a parametric model to control the influence of adjacent regions. Alternatively, they might instead be estimated using statistical learning approaches, e.g., lasso-type estimators under the assumption of a certain degree of sparsity. This section of the paper goes on to discuss two estimation strategies, namely nonlinear least squares estimators and estimators that are based on the maximum-likelihood principle.

\subsection{Choice of the Weight Matrices}

There is a great flexibility in the choice of the weight matrices (see \citealt{Getis09} for an overview). In practice, these are usually dependent upon additional parameters and spatial locations.
Frequently, it is assumed that $\xmat{W}_1 = \rho \xmat{W}^{*}_1$ and $\xmat{W}_2 = \lambda \xmat{W}^{*}_2$ with the predefined, known matrices $\xmat{W}^{*}_1$ and $\xmat{W}^{*}_2$. That is, $\xmat{W}^{*}_1$ and $\xmat{W}^{*}_2$ describe the structure of the spatial dependence, with the weights as a multiple of these specific matrices. In settings such as these, it is easy to test whether a random process exhibits such a spatial dependence, by testing the parameters $\rho$ and $\lambda$. As with time-series GARCH models, $\rho$ measures the extent to which a volatility shock in one region spills over to neighboring regions, while $\rho + \lambda$ gives an impression how fast this effect will fade out in space (see, e.g., \citealt{Campbell97}). A more general approach can be obtained by choosing $\xmat{W}_{\cdot} = \text{diag}(\rho_1, \ldots, \rho_1, \ldots, \rho_k, \ldots, \rho_k) {\xmat{W}_{\cdot}^{*}}$ as the weights. Here, different areas are weighted in different ways. For instance, all counties of state $i$ are weighted by $\rho_i$, while counties of another state, $j$, get a different weighting factor, $\rho_j$. Alternatively, $\xmat{W}_{\cdot}^{*}$ could be chosen as $(K_\theta(\xvec{s}_i - \xvec{s}_j))_{i,j = 1, \ldots, n}$ with a known function, $K$. In this case, the spatial correlation depends on the distance between two locations. For instance, inverse distance weighting schemes $K(\xvec{x}) = ||\xvec{x}||^{-k}$ with $k$ being estimated, or anisotropic weighting schemes dependent upon the bearing between two locations.

\subsection{Parameter Estimation}

Below, we assume that the weight matrices have the structure
\begin{equation}\label{weight}
\xmat{W}_1 = \rho \xmat{W}_1^*, \xmat{W}_2 = \lambda \xmat{W}_2^*, \xvec{\alpha} = \alpha \xvec{1} .
\end{equation}
Thus, the model has three parameters to be estimated, $\xvec{\vartheta} = (\rho, \lambda, \alpha)^\prime$. Let $\xvec{\vartheta}_0$ denote the true parameters.

\subsubsection{Nonlinear Least Squares Estimators}

One possible method for estimating the parameters is the nonlinear least squares approach (NLSE). We initially limit ourselves to GARCH Type I processes.

Squaring the components of \eqref{eq:initial} and taking the logarithms, we get that for $i=1, \ldots, n$
\begin{eqnarray*}
\mbox{log}(Y(\xvec{s}_i)^2) & = & \mbox{log}(h(\xvec{s}_i)) + \mbox{log}(\varepsilon(\xvec{s}_i)^2)\\
& = & E( \mbox{log}(\varepsilon(\xvec{s}_i)^2) ) + \mbox{log}(h(\xvec{s}_i)) + \eta(\xvec{s}_i)
\end{eqnarray*}
with $\eta(\xvec{s}_i) = \mbox{log}(\varepsilon(\xvec{s}_i)^2) - E( \mbox{log}(\varepsilon(\xvec{s}_i)^2) )$. Now $\eta(\xvec{s}_i), i=1,...,n$ is a white noise process. Moreover, it follows with $\tau(x) = f(\mbox{exp}(x))$ that
\begin{eqnarray*}
\xvec{F} & = & ( \tau(\mbox{log}(h(\xvec{s}_i)) )_{i=1,...,n} = (\xmat{I} -  \lambda \xmat{W}_2^*)^{-1} (
\alpha \xvec{1} + \rho \xmat{W}_1^* \xvec{\gamma}(\xvec{Y}^{(2)} ) ) \\
& = &  (\xmat{I} -  \lambda \xmat{W}_2^*)^{-1} (\alpha \xvec{1} + \rho \xmat{W}_1^* \tilde{\xvec{\gamma}}(\mbox{log}(\xvec{Y}^{(2)}) ) )
\end{eqnarray*}
where $ \tilde{\xvec{\gamma}}(\xvec{x}) = ( \gamma_i(\mbox{exp}(x_1),..., \mbox{exp}(x_n)) )_{i=1,...,n}$.

\noindent Now, let $( c_i(\lambda) )_{i=1,...,n} = (\xmat{I} -  \lambda \xmat{W}_2^*)^{-1} \xvec{1}$ and $( \xvec{d}_i(\lambda)^\prime )_{i=1,...,n} = (\xmat{I} -  \lambda \xmat{W}_2^*)^{-1} \xmat{W}_1^*$.  In order to denote the dependence on $\xvec{\vartheta}$ we write $h_{\xvec{\vartheta}}(\xvec{s}_i)$, $i=1,...,n$. Then
\[ \mbox{log}(h_{\xvec{\vartheta}}(\xvec{s}_i)) = \tau^{-1}(\alpha \, c_i(\lambda) + \rho \, \xvec{d}_i(\lambda)^\prime \,  \tilde{\xvec{\gamma}}(\mbox{log}(\xvec{Y}^{(2)} ) ) ) . \]

Here, we assume that $c = E( \mbox{log}(\varepsilon(\xvec{s}_i)^2) )$ is a known quantity. Using $H_i = \mbox{log}(Y(\xvec{s}_i)^2) - E( \mbox{log}(\varepsilon(\xvec{s}_i)^2) )$ and $\xvec{H} = ( H_i )_{i=1,...,n}$ the estimators of the parameters $\alpha$, $\lambda$, and $\rho$ are obtained by minimizing the nonlinear sum of squares
\[
\sum_{i=1}^n \left( H_i - \mbox{log}(h_{\xvec{\vartheta}}(\xvec{s}_i) \right)^2 = \sum_{i=1}^n \left( H_i - \tau^{-1}(\alpha \, c_i(\lambda) + \rho \, \xvec{d}_i(\lambda)^\prime \,  \tilde{\xvec{\gamma}}(\xvec{H} + c \xvec{1} ) ) \right)^2 \]
with respect to $\xvec{\vartheta}$.

Although $\tau^{-1}$ is a known function, this minimization problem is complex. Thus, we will impose further assumptions which are fulfilled for all relevant special cases. We will suppose that $\xvec{\gamma}(\xvec{x}) = ( \gamma( x_i ) )$ with a known function $\gamma$. Consequently, $\tilde{\xvec{\gamma}}(\xvec{x}) = ( \tilde{\gamma}(x_i) )_{i=1,...,n}$ with $\tilde{\gamma}(x) = \gamma( \mbox{exp}(x) )$, which leads to the easier minimization of
\[ Q_n(\xvec{\vartheta} ) = \frac{1}{n} \; \sum_{i=1}^n \left( H_i - \tau^{-1}\left(\alpha \, c_i(\lambda) + \rho \, \xvec{d}_i(\lambda)^\prime \,  \left( \tilde{\gamma}( H_v + c )\right) \right)_{v=1,...,n}
 \right)^2 . \]
Note that since the $(i,i)$-th element of $\xmat{W}_1^*$ is zero it follows that \linebreak $\xvec{d}_i(\lambda)^\prime \,  \left( \tilde{\gamma}( H_v + c ) \right)_{v=1,...,n}$ is independent of $H_i$.
Minimization problems of that type have been studied in detail in, e.g., \cite{Amemiya85,Potscher97,Newey94}. There, sufficient conditions are given for the consistence and asymptotic normality of the resulting estimators under various conditions. Note that, in the present case, $\{ H_i \}$ is a strictly stationary process.

To prove the consistency of the NLSE, we can apply Lemma 3 of \cite{Amemiya85}, as we will briefly sketch below. It is sufficient to show that there is a function $Q_0(\xvec{\vartheta})$ such that $Q_0(\xvec{\vartheta})$ is uniquely minimized at $\xvec{\vartheta}_0$, that $\Theta$ is compact, that $Q_0(\xvec{\vartheta})$ is continuous, and that $Q_n(\xvec{\vartheta})$ converges uniformly in probability to $Q_0(\xvec{\vartheta})$. Recall that $H_i = \log(h_{\xvec{\vartheta}_0}(\xvec{s}_i)) + \log(\varepsilon(\xvec{s}_i)^2) - E(\log(\varepsilon(\xvec{s}_i))) = \eta(\xvec{s}_i) + \log(h_{\xvec{\vartheta}_0}(\xvec{s}_i))$. Thus, in the present case, the function $Q_0(\xvec{\vartheta})$ can be found by splitting $Q_n(\xvec{\vartheta})$ into eight parts, as follows
\[ Q_n(\xvec{\vartheta}) = I_n + II_n + III_n + IV_n + V_n + VI_n + VII_n + VIII_n \]
with
\begin{small}
\begin{eqnarray*}
  I_n    & = & \frac{1}{n} \sum_{i=1}^n \eta(\xvec{s}_i)^2, \quad II_n = \frac{1}{n} \sum_{i=1}^n \left( E( \mbox{log}(h_{\xvec{\vartheta}}(\xvec{s}_i)) - E(\mbox{log}(h_{\xvec{\vartheta}_0}(\xvec{s}_i)) ) \right)^2 , \\
  III_n  & = & \frac{1}{n} \sum_{i=1}^n \left( \mbox{log}(h_{\xvec{\vartheta}_0}(\xvec{s}_i)) - E( \mbox{log}(h_{\xvec{\vartheta}_0}(\xvec{s}_i))) \right)^2 \\
  IV_n   & = & \frac{1}{n} \sum_{i=1}^n \left( \mbox{log}(h_{\xvec{\vartheta}}(\xvec{s}_i)) - E( \mbox{log}(h_{\xvec{\vartheta}}(\xvec{s}_i))) \right)^2 \\
  V_n    & = & \frac{2}{n} \sum_{i=1}^n \eta(\xvec{s}_i)  \left( \mbox{log}(h_{\xvec{\vartheta}_0}(\xvec{s}_i)) - \mbox{log}(h_{\xvec{\vartheta}}(\xvec{s}_i) ) \right)  \\
  VI_n   & = & \frac{2}{n} \sum_{i=1}^n (E(\log(h_{\xvec{\vartheta}_0}(\xvec{s}_i))) - E(\log(h_{\xvec{\vartheta}}(\xvec{s}_i)))) (\log(h_{\xvec{\vartheta}_0}(\xvec{s}_i)) - E(\log h_{\xvec{\vartheta}_0}(\xvec{s}_i))) ,  \\
  VII_n  & = & -\frac{2}{n} \sum_{i=1}^n (E(\log(h_{\xvec{\vartheta}_0}(\xvec{s}_i))) - E(\log(h_{\xvec{\vartheta}}(\xvec{s}_i)))) (\log(h_{\xvec{\vartheta}}(\xvec{s}_i)) - E(\log h_{\xvec{\vartheta}}(\xvec{s}_i))) ,  \\
  VIII_n & = & -\frac{2}{n} \sum_{i=1}^n (\log(h_{\xvec{\vartheta}_0}(\xvec{s}_i)) - E(\log h_{\xvec{\vartheta}_0}(\xvec{s}_i))) (\log(h_{\xvec{\vartheta}}(\xvec{s}_i)) - E(\log h_{\xvec{\vartheta}}(\xvec{s}_i))).  \\
\end{eqnarray*}
\end{small}
Assuming $\{ \varepsilon(\xvec{s}_i) \}$ to be an independent and identically distributed random process with existing moment $E( \mbox{log}(\varepsilon(\xvec{s}_1)^2) )$, the weak law of large numbers (WLLN) implies that $I_n$ converges in probability to $E( \mbox{log}(\varepsilon(\xvec{s}_1)^2) )$.
The quantity $II_n$ is zero for $\xvec{\vartheta} = \xvec{\vartheta}_0$ and we must ensure that this is the only zero. The limit of this quantity determines $Q_0(\xvec{\vartheta})$ up to a constant. The quantity $III_n$ does not depend on $\xvec{\vartheta}$ at all.

For the spGARCH model, for instance, it is assumed that $\varepsilon(\xvec{s}_i)$ has a bounded support (see \cite{Otto18_spARCH} for spARCH). This assumption is sufficient to ensure the existence of a strictly stationary solution (see Theorem 1). Note that this implies that $Y(\xvec{s}_i)$ is uniformly bounded.
Since $f(x) = x$ we get that  $\tau^{-1} = \mbox{log}(x)$ and $\mbox{log}(h_{\xvec{\vartheta}}(\xvec{s}_i) =\mbox{log}(\alpha c_i(\lambda + \rho \xvec{d}_i(\lambda)^\prime \xvec{Y}^{(2)}))$.
Then it follows with $\kappa = \rho/\alpha$ that
\[ IV_n =  \frac{1}{n} \sum_{i=1}^n \left( \mbox{log}\left(1 + \kappa \frac{\xvec{d}_i(\lambda)^\prime}{c_i(\lambda)} \xvec{Y}^{(2)}   \right) - E\left( \mbox{log}\left(1 + \kappa \frac{\xvec{d}_i(\lambda)^\prime }{c_i(\lambda)} \xvec{Y}^{(2)} \right)  \right) \right) \]
\[ =  \kappa \; \frac{1}{n} \sum_{i=1}^n \left(  \frac{\xvec{d}_i(\lambda)^\prime}{c_i(\lambda)} \left( \xvec{Y}^{(2)}   - E( \xvec{Y}^{(2)} )  \right) \right) + o_p(1)  \]
which converges to zero in probability, if the second moments of $\gamma(Y(\xvec{s}_i)^2), i=1,...,n$ exist and $\frac{\xvec{1}^\prime \tilde{\xmat{W}} Cov(\xvec{Y}^{(2)}) \tilde{\xmat{W}}^\prime \xvec{1}}{n^2} \rightarrow 0$ with $\tilde{\xmat{W}} = (\frac{\xvec{d}_i(\lambda)^\prime}{c_i(\lambda)} )_{i=1,...,n}$. Using Lemma 2.9 of \cite{Newey94} the convergence can be shown to be uniform.

Since $\{ \eta(\xvec{s}_i) \}$ follows a WLLN and $Y(\xvec{s}_i)$ is bounded the quantity $V_n$ follows a WLLN as well.

As a consequence, $VI_n$ to $VIII_n$ do also approach zero, because of the Cauchy-Schwarz inequality, e.g., $| VI_n | \le \sqrt{ II_n  III_n}$.

\subsubsection{Maximum Likelihood Estimator}

Another approach to estimating the parameters of our model uses the maximum-likelihood principle. Below, we initially confine ourselves to GARCH Type II processes.

Our first aim is to derive the likelihood function. From \eqref{eq:unified}, we can obtain in general that
\[ \xvec{F} = (\xmat{I} - \xmat{W}_2)^{-1} (\xvec{\alpha} + \xmat{W}_1 \xvec{g} (\xvec{\varepsilon}) ) . \]
For the $i$-th location, $f(h(\xvec{s}_i))$ is equal to $c_i + \xvec{d}_i^\prime \xvec{g}(\xvec{\varepsilon})$. If $f$ is invertible, we obtain that $h(\xvec{s}_i) = f^{-1}( c_i + \xvec{d}_i^\prime \xvec{g}(\xvec{\varepsilon}));$ thus,
\[ Y(\xvec{s}_i) = \varepsilon(\xvec{s}_i)  \sqrt{ f^{-1}( c_i + \xvec{d}_i^\prime \xvec{g}(\xvec{\varepsilon}))} \]
if $h(\xvec{s}_i) \ge 0$. Suppose that $Y(\xvec{s}_i)$ is a continuously differentiable function of $\xvec{\varepsilon}$. The $(i,j)$th element $J_{ij}$ of the Jacobian matrix $J  = \left( \frac{\partial Y(\xvec{s}_i)}{\partial \varepsilon(\xvec{s}_j)} \right)_{i,j=1,\ldots,n}$ is given by
\begin{footnotesize}
\begin{equation*}
  \frac{1}{2} \frac{1}{\sqrt{f^{-1}( c_i + \xvec{d}_i^\prime \xvec{g}(\xvec{\varepsilon}))}} \; \frac{1}{f^\prime(f^{-1}( c_i + \xvec{d}_i^\prime \xvec{g}(\xvec{\varepsilon})))} \;
\xvec{d}_i^\prime \frac{\partial \xvec{g}(\xvec{\varepsilon})}{\partial \varepsilon(\xvec{s}_j)} \; \varepsilon(\xvec{s}_i) + \sqrt{f^{-1}( c_i + \xvec{d}_i^\prime \xvec{g} (\xvec{\varepsilon}))} I_{\{i\}}(j)
 .
\end{equation*}
\end{footnotesize}
The indicator function on a set $A$ is denoted by $I_{\{A\}}(x)$. Following the inverse function theorem, the inverse function exists in a neighborhood of all points, for which the determinant of the Jacobian matrix is not equal to zero. In that case, we have $\xvec{\varepsilon} = \xi(\xvec{Y})$.

In the examples that we have considered above, the inverse function $\xi$ is easily obtained. For instance, the inverse function for the E-spGARCH model is given in (\ref{eq:inverse_f}). For spGARCH and H-spGARCH models, it holds that $\varepsilon(\xvec{s}_i) = Y(\xvec{s}_i)/\sqrt{h(\xvec{s}_i)}$, where $h$ is a function of $\xvec{Y}$ (i.e., $\xi(\xvec{Y}) = \xvec{Y} \circ (1/\sqrt{h(\xvec{s}_i)} )_{i=1,\ldots,n}$).
If the inverse function exists, then the transformation rule can be applied, and it leads to the likelihood function
\begin{equation}\label{eq:likelihood}
f_{\xvec{Y}}(\xvec{y}) = f_{\xvec{\varepsilon}}(\xi(\xvec{y})) \frac{1}{\mbox{det}(J(\xi(\xvec{y})))} .
\end{equation}
For the four examples considered above, it holds that
\[ J_{ij} =  \frac{1}{2} \; \frac{\varepsilon(\xvec{s}_i)}{ \sqrt{h(\xvec{s}_i)}} \; \frac{\partial h(\xvec{s}_i)}{\partial \varepsilon(\xvec{s}_j)} + \sqrt{h(\xvec{s}_i)} I_{\{i\}}(j) . \]
In the Appendix, we discuss how the partial derivatives $\frac{\partial h(\xvec{s}_i)}{\partial \varepsilon(\xvec{s}_j)}$ can be obtained for the models mentioned above.

The main problem with calculating the likelihood function is the determinant $\mbox{det}(J_{ij})$. Unfortunately, there is no explicit expression for this quantity due to the general formulation of the process. For that reason, it is difficult to analyze the likelihood function and to make statements about the properties of the corresponding ML estimators. However, it is possible to use any nonlinear optimization algorithm, such as \cite{RSolnp}, to determine the ML estimators numerically. Since all restrictions on the parameters are linear, any Kuhn-Karush-Tucker point would coincide with the global optimum. Thus, we have analyzed the performance and properties of both estimation procedures in a series of Monte Carlo simulation studies, which are explained below. In particular, we discuss their computational implementation and the insights that we gain from these studies.

\section{Computational Implementation and Simulation Studies}

Over the course of the following paragraphs, our aim is to compare the considered spatial GARCH-type models. We ask, how well does the spGARCH model fit the data if it follows an E-spGARCH, log-spGARCH, or H-spGARCH model in reality?  What do we lose when choosing the wrong model? And, how different are the models from one another?

Here, we assume the simple parametric setting given by \eqref{weight}.  The weighting matrices $\xmat{W}^*_1$ and $\xmat{W}^*_2$ were set as row-standardized Rook and Queen contiguity matrices, respectively. Furthermore, the upper diagonal elements were set to zero to avoid negative values of $h(\xvec{s}_i)$, which may arise while simulating spGARCH processes. Thus, the processes can be interpreted as directional spatial processes with the origin in the upper right corner. The simulation study is performed on a $15 \times 15$ spatial unit grid (i.e., $D_{\xvec{s}} = \{\xvec{s} = (s_1, s_2)' \in \xset{Z}^2 : 1 \leq s_1, s_2 \leq 15 \}$), resulting in $n = 15^2 = 225$ observations, with $m = 1000$ replications.

We computationally implemented the unified approach in the \textbf{R}-package \texttt{spGARCH} from version \texttt{0.2.0}. The package provides several main functions; foremost it provides functions for parameter estimation for both spatial GARCH-type and mixed spatial autoregressive models that have spatial GARCH-type residuals, but also contains a random number generator for spatial GARCH-type models.
All examples are included if the \texttt{type} is set to \texttt{spGARCH}, \texttt{e-spGARCH}, or \texttt{log-spGARCH}. Regarding the error distribution $f_\varepsilon$, the package implements standard normally distributed random errors for the E-spGARCH, log-spGARCH, and spGARCH models. In the latter case, the normal distribution is truncated if there is no permutation such that the weighting matrix is triangular. Please refer to \cite{Otto19_RJournal} for a detailed description of the package.

\begin{figure}
  \centering
  \small{Nonlinear least-squares approach}\\
  \includegraphics[width=0.6\textwidth]{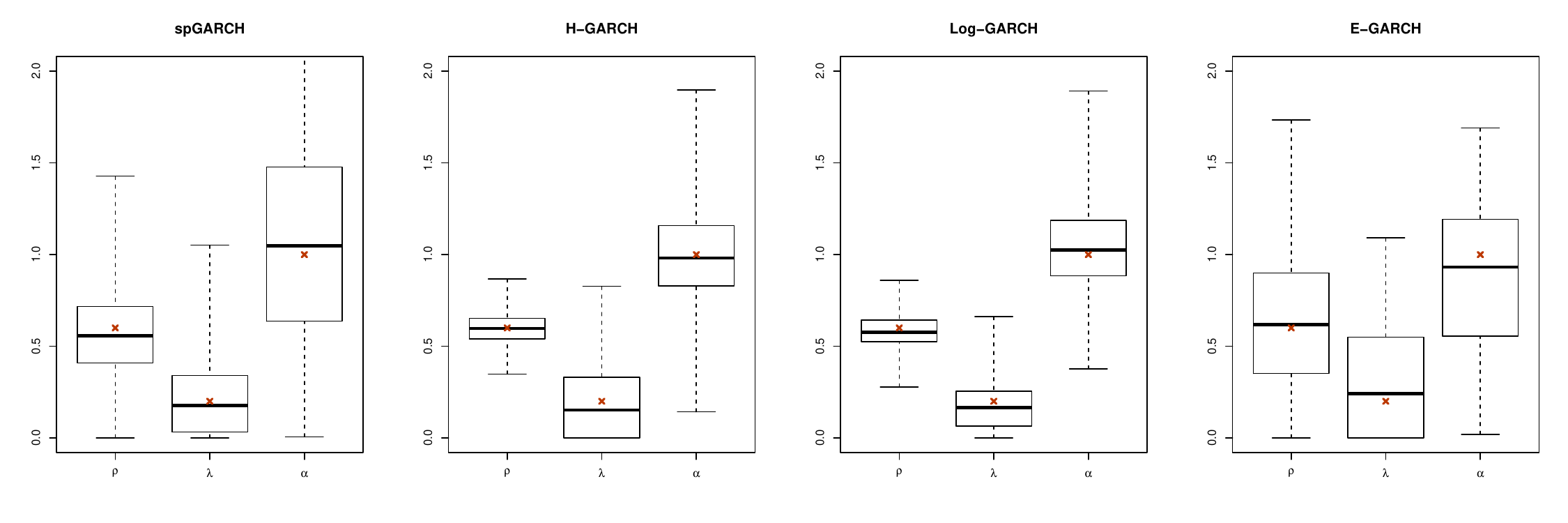}\\
  \small{Maximum-likelihood approach}\\
  \includegraphics[width=0.6\textwidth]{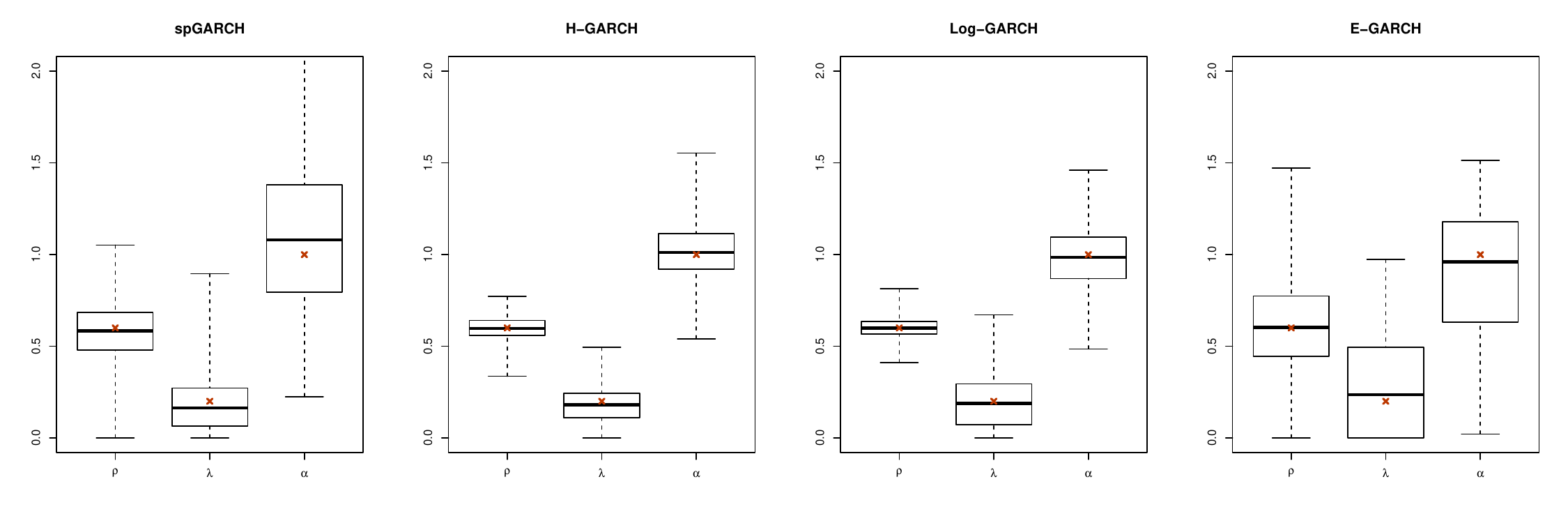}
  \caption{Boxplot of the estimated parameters for 1000 replications to illustrate the performance of the proposed estimators (first row: nonlinear least-squares estimators, second row: maximum-likelihood estimators). The true underlying values are indicated using an asterisk for each parameter.}\label{fig:estimation}
\end{figure}

For model selection, criteria based on the maximum value of the likelihood function, or on the residual's variance, can be used, such as the Akaike or Schwarz information criteria. Because the Bayesian information criterion (BIC) would tend to select a rather parsimonious model, we keep the number of unknown parameters constant for all models (i.e., $\Theta = 0.5$, $\zeta = 0$, and $b = 2$) that are assumed to be known for estimation. Thus, we analyze the performance of the model selection directly from the attained maximum likelihood.

The average values of logarithmic likelihood are reported in Table \ref{table:selection} and the percentages of cases in which the models were selected by this procedure are given in parentheses. The correct model was selected in the majority of cases. Generally, we observe three different groups of models. While the H-spGARCH and log-spGARCH models lead to similar results and values of the log-likelihood, the spGARCH and E-spGARCH models are different.

In contrast to a spGARCH process with additive dynamics, the spatial E-spGARCH process has a multiplicative structure. Thus, it is important that the correct model is selected. In fact, the average values of log-likelihood for the spGARCH and E-spGARCH models are not similar. On the contrary, we observe similar values for the H-spGARCH and log-spGARCH approaches. As the hybrid model would be identical to the log-spGARCH process, if $\xmat{W}^*_1 = \xmat{W}^*_2$ then only a few links in the weighting matrix distinguish between these two models. Thus, these models are more similar than the spGARCH and E-spGARCH models.

\begin{table}
  \caption{The average value of the logarithmic likelihood function and the percentage of selections within 1000 trials per row. The rows give the simulated models and the columns give the estimated models, so that correctly selected models can be read on the diagonal. The largest average log-likelihood and percentage of selections are printed in bold in each row.}\label{table:selection}
  \centering
  \vspace*{.3cm}
  \hspace*{-1.3cm}
  \begin{scriptsize}
  \begin{tabular}{ >{\centering}m{1.8cm}  cc cc cc cc}
     %\hline
      & \multicolumn{8}{c}{Estimated model} \\
     Simulated model & \multicolumn{2}{c}{spGARCH} & \multicolumn{2}{c}{H-spGARCH} & \multicolumn{2}{c}{Log-spGARCH} & \multicolumn{2}{c}{E-spGARCH} \\
     \hline
     spGARCH       & \textbf{-455.722} & \textbf{(93.9\%)}  &         -463.629  &         (3.1\%)   &         -464.023  &         (2.5\%)    &         -474.053   &         (0.5\%)    \\
     H-spGARCH     &         -472.452  &         (0.0\%)    & \textbf{-451.418} & \textbf{(93.7\%)} &         -455.530  &         (6.3\%)    &         -535.889   &         (0.0\%)    \\
     Log-spGARCH   &         -378.353  &         (0.0\%)    &         -362.973  &         (0.3\%)   & \textbf{-352.104} & \textbf{(99.7\%)}  &         -393.634   &         (0.0\%)    \\
     E-spGARCH     &         -461.607  &         (3.3\%)    &         -461.850  &         (2.2\%)   &         -462.029  &         (2.1\%)    & \textbf{-458.491}  & \textbf{(92.4\%)}  \\
     \hline
   \end{tabular}
   \end{scriptsize}
\end{table}

For the correct models (i.e., all combinations on the diagonal of Table \ref{table:selection}), we have further analyzed the performance of the estimators. The estimated coefficients are depicted by a series of boxplots in Figure \ref{fig:estimation} for both estimation strategies. The true values $\rho = 0.5$, $\lambda = 0.4$, and $\alpha = 1$ are shown using asterisks. Due to the similarity of the H-spGARCH and the log-spGARCH, it is not surprising that their estimated values are very similar. Hence, these two fitted models would generate values that are very close to the true values, even if the wrong model is selected. The estimators are also unbiased for the spGARCH and E-spGARCH, although slightly less precise. Generally, we see that both estimators, the NLSE and MLE, perform equally well.

\section{Real-World Application: Condominium Prices in Berlin}

In markets that are constrained in space, one can typically expect to find locally varying risks. Typical examples of such markets are real estate and labor. For the former, the property prices are highly dependent upon the location of the real estate and prices in the surrounding areas. Similarly, for the latter, this market is also often constrained in space due to the limited mobility of laborers. On the one hand, we observe conditional mean levels that vary in space, so-called spatial clusters. On the other hand, we may also expect to find locally varying risks relating to price, which can be considered as local volatility clusters. The proposed spatial GARCH-type models are capable of capturing such spatial dependencies in the conditional variance. This motivates why we consider condominium prices at a fine spatial scale. In particular, we will analyze the relative changes in Box-Cox transformed prices from 1995--2014 across all Berlin ZIP-Code regions (i.e., $n = 190$). The data are depicted on the left in Figure \ref{fig:example1}. The sample mean for these price changes is 0.8103 with a median of 0.6965. In total, the price changes range from -2.5650 to 7.3131. We can observe a spatial cluster of positive values in the north-western ZIP-Code regions. In addition, we depict the row-standardized contiguity matrix $\xmat{W}_1^* = \xmat{W}_2^* = \xmat{W}^*$ on the right in Figure \ref{fig:example1}.
\begin{figure}
  \centering
  \includegraphics[width=0.4\textwidth]{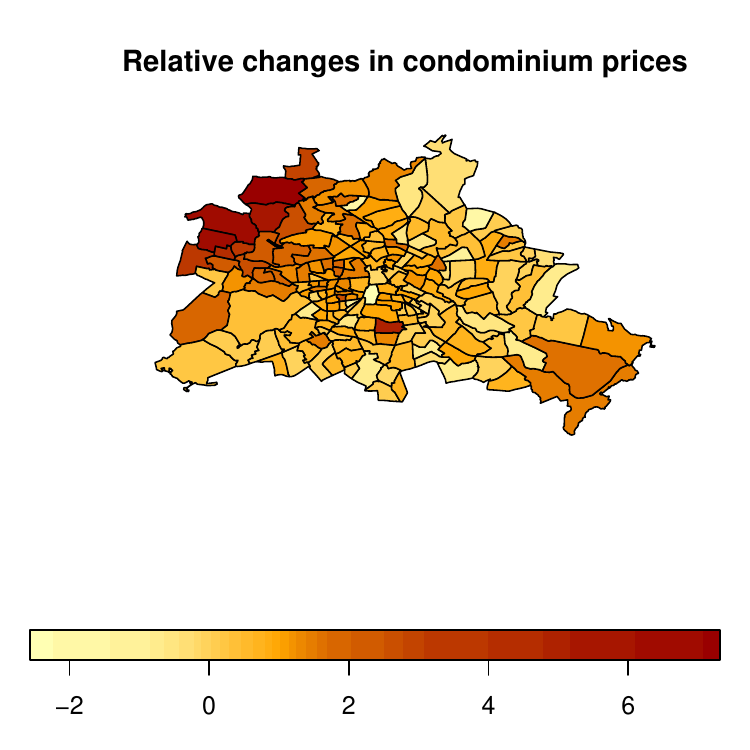}
  \includegraphics[width=0.4\textwidth]{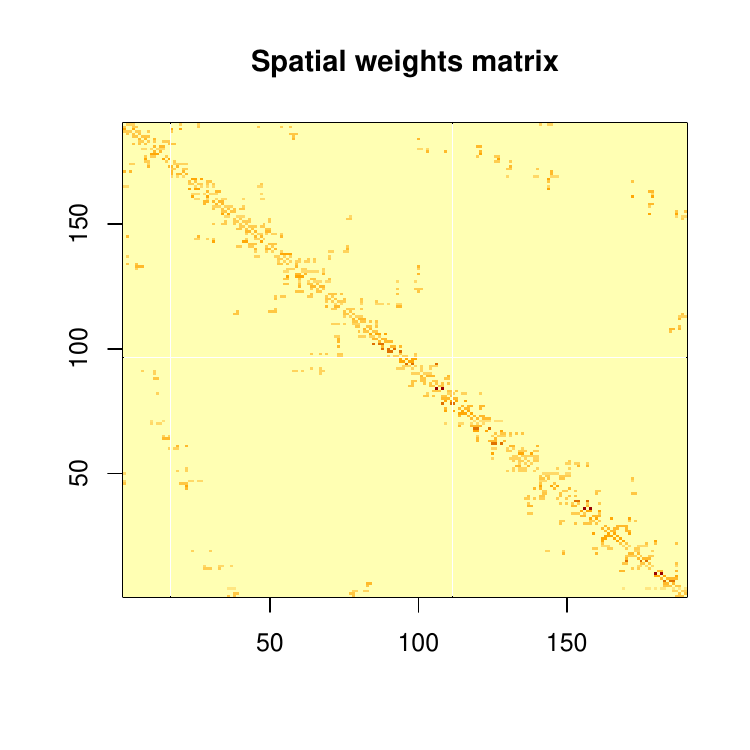}
  \caption{Relative changes (in per cent) of the Box-Cox transformed prices for all Berlin ZIP-Code regions, along with the applied spatial weight matrix, which coincides with a classical row-standardized first-order contiguity matrix.}\label{fig:example1}
\end{figure}

To account for dependence in the local means, we first estimate the coefficients of a classical spatial autoregressive model, briefly SAR (see, e.g. \citealt{Halleck15,Lee04}), and then fitting the spGARCH, E-spGARCH, and log-spGARCH models to the residuals of this autoregressive model. Thus, the GARCH parameters can be interpreted as local model uncertainties of the linear spatial autoregressive model. Moreover, these GARCH effects cover latent variables, because natural candidates of regressors -- such as the average income per capita, the job market situation, property taxes, public transport connections and cultural offerings -- cannot be included due to the fine spatial scale of ZIP-Code levels. There is no reliable way to associate quantities like personal orhousehold income with ZIP-Code areas. Moreover, public services are very similar in Berlin across all areas, such that the differences in prices of real estate are rather more likely to be associated with lifestyle factors such as whether or not an area is `trending'. All of these latent variables can be captured by spatial GARCH effects.  Thus, we estimated a purely autoregressive model without any regressors. Moreover, we assume $t$-distributed residuals for the ML approach, since the changes in the house prices and the residuals of the SAR model follow a heavy-tailed distribution, as it is also often observed for financial data. The distribution $f_\varepsilon$ is chosen from a set of heavy-tailed distributions by minimizing the BIC. More precisely, we consider $t$-distributions with between one (Cauchy distribution) and nine degrees of freedom, and we consider the standard normal distribution as potential error distributions. As in the MC studies, we have chosen the additional parameters of the Log-spGARCH and E-spGARCH to be constant, namely $\Theta = 1$ and $b = 2$. In Table \ref{table:results_ex1}, we summarize the fitted parameters of the best-fitting models. Since, we have chosen $\xmat{W}_1^* = \xmat{W}_2^* = \xmat{W}^*$, the fitted log-spGARCH model coincides with the hybrid model. According to the BIC, the E-spGARCH model fits the residual process the best. While the spGARCH and Log-spGARCH attain their minimal BIC for a $t$-distribution with 3 degrees of freedom, we get 7 degrees of freedom for the E-spGARCH. Thus, the E-spGARCH tends to model heavier tails than the other two models. Obviously, positive spatial dependence exists in the conditional heteroscedasticity of the residuals because all parameters are positive. This implies that clusters of higher model uncertainties and, therefore, larger prediction intervals, do exist. Moreover, these uncertainties spill over to the neighboring ZIP-Code regions. More precisely, we would observe multiplicative dynamics of the conditional variances for the E-spGARCH attempt. Because we would not expect asymmetry effects, $\zeta$ has been set to zero. Eventually, we report a robustified version of Moran's $I$ based on percentage bends (see \citealt{Wilcox94}) as a measure for the spatial dependence of the residuals and squared residuals. All of them do not significantly differ from zero.

\begin{table}
    \caption{Estimated parameters of the spGARCH model for the residuals of a spatial autoregressive model, where the dependent variables are the changes in the condominium prices in Berlin. All standard errors are given in parentheses.}\label{table:results_ex1}
        \centering
        {\scriptsize
    \begin{tabular}{l cc cc cc}
      % \hline
                                                & \multicolumn{2}{c}{spGARCH model}  & \multicolumn{2}{c}{Log-spGARCH model}    & \multicolumn{2}{c}{E-spGARCH model}    \\
      Parameter                                 & \multicolumn{2}{c}{Estimate}       & \multicolumn{2}{c}{Estimate}           & \multicolumn{2}{c}{Estimate}         \\
      \hline
      \emph{Mean equation}                      &          &                         &          &                             &          &                           \\
      $\quad$ $\mu$                             &                          \multicolumn{6}{c}{0.2002 (0.0813)}                                                       \\
      $\quad$ $\gamma$                          &                          \multicolumn{6}{c}{0.7456 (0.0588)}                                                       \\
      \emph{Residuals process}                  &          &                         &          &                             &          &                           \\
      $\quad$ $\alpha$                          & 0.1757   & (0.2181)                & 0.0000   & (0.2157)                    & 2.8703   & (0.7852)                  \\
      $\quad$ $\rho$                            & 0.0609   & (0.0579)                & 0.0507   & (0.1791)                    & 1.5339   & (0.7521)                  \\
      $\quad$ $\lambda$                         & 0.2093   & (0.8686)                & 0.9646   & (0.0784)                    & 3.0380   & (1.2054)                  \\
%       $\quad$ $b$                               &    -     &    -                    & 0.0515   & (0.6834)                    &    -     &    -                      \\
%       $\quad$ $\Theta$                          &    -     &    -                    &   -      &   -                         & 2.3543   & (1.1527)                  \\
%     $\quad$ $\zeta$                           &    -     &    -                    &   -      &   -                         &          &                           \\
      \emph{Summary statistics}                 &          &                         &          &                             &          &                           \\
      $\quad$ Degrees of freedom $f_\varepsilon$ & \multicolumn{2}{c}{3}             & \multicolumn{2}{c}{3}                  & \multicolumn{2}{c}{7}                \\
      $\quad$ LL                                & \multicolumn{2}{c}{-217.1787}      & \multicolumn{2}{c}{-218.6777}          & \multicolumn{2}{c}{-185.6466}        \\
      $\quad$ BIC                               & \multicolumn{2}{c}{450.0985}       & \multicolumn{2}{c}{453.0965}           & \multicolumn{2}{c}{387.0343}         \\
%      $\quad$ $C$      res. ($p$-value)         &  1.0607  & (0.2491)                &  1.0517  & (0.3148)                    & 1.0863   & (0.0681)          \\
%      $\quad$ $C$ squ. res. ($p$-value)         &  1.0557  & (0.4893)                &  1.0174  & (0.8204)                    & 0.8693   & (0.0533)          \\
      $\quad$ $I_r$ res. ($p$-value)              &  -0.0725  & (0.3201)                &  -0.0724  & (0.3211)                    & -0.0706   & (0.3325)          \\
      $\quad$ $I_r$ squ. res. ($p$-value)         &  -0.0042  & (0.9539)                &   0.0072  & (0.8204)                    &  0.0605   & (0.9210)          \\
    \end{tabular}
    }
\end{table}

\section{Discussion and Conclusions}\label{sec:conclusion}

Recently, a few papers have introduced spatial ARCH and GARCH-type models that allow the modelling of an instantaneous spatial autoregressive dependence of heteroscedasticity. In this paper, we propose a unified spatial GARCH model that covers all previous approaches. Alongside these spatial ARCH-type processes, we introduce a completely novel spGARCH and E-spGARCH model. Due to the flexible definition of the model as a set of functions, we can derive a common estimation and a model selection strategy for all these spatial GARCH-type models, that is based on nonlinear least squares and a maximum-likelihood approach.

In addition, we want to stress that the dependence structure does not necessarily have to be interpreted in a spatial sense. Thus, we briefly discuss a further example below, on which the ``spatial'' proximity could also be defined as the edges of networks. In such cases,  $\xmat{W}_1$ and $\xmat{W}_2$ would be interpreted as adjacency matrices. For instance, one might consider the financial returns of several stocks as a network, where the only assets that are connected are those that are correlated  above a certain threshold. This could create a financial network, as shown in Figure \ref{fig:example2}. Thus, spGARCH models can be used to analyze various forms of information, whether that might be volatility, risk, or spillovers from one stock to another, if these assets are close to one another within a certain network. In future research, attempts for modelling volatility clusters within networks, using spatial GARCH models, should be analyzed in greater detail.
\begin{figure}
  \centering
  \includegraphics[width=0.4\textwidth]{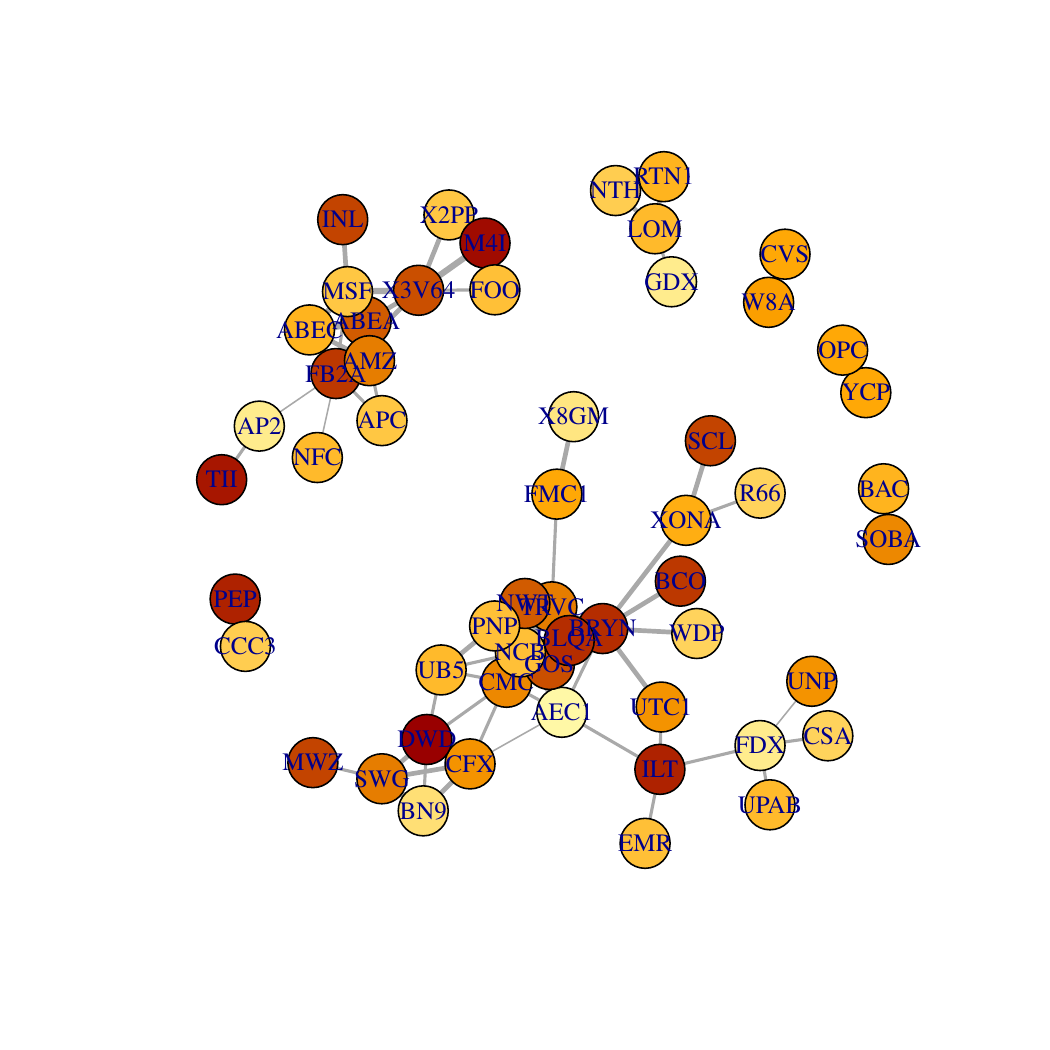}
  \caption{Financial network of selected stocks of the S\&P 500, where the color of the nodes denotes the annual returns in 2017 with darker colors indicating higher returns.}\label{fig:example2}
\end{figure}

In the second part of this paper, we evaluated the performance of both the parameter estimation and model selection procedures, using a Monte Carlo simulation study. As this unified approach is already available in the \textbf{R}-package \texttt{spGARCH}, we briefly sketched its computational implementation. Eventually, the use of the model was demonstrated through an empirical example. More precisely, this paper has shown how the model uncertainties of local price changes in the real estate market in Berlin can be described using an E-spGARCH model. Moreover, the residuals of fitted spatial GARCH-type models seem to follow a heavy-tailed distribution. Though all proposed models are uncorrelated and have a zero mean, potential interactions between the error process and the mean equation should be analyzed in greater detail in future research.

In this paper, we have assumed that suitable functions of the unified model framework are known. Hence, it is possible to maximize certain goodness-of-fit criteria in order to obtain the best-fitting model. However, these functions can also be estimated using a nonparametric approach; for instance by penalized or classical B-splines. This approach will be the subject of a forthcoming paper.

\section{Appendix}
\section*{Proofs and Further Results on the Stochastic Properties of the Unified Approach}

Initially, we provide the proofs of the Theorems 1 - 3 and the Corollaries 1 and 2.

\begin{proof}[Theorem \ref{solution}]
Inserting \eqref{eq:initial} into \eqref{eq:unified}, we get that

\[ \xvec{F}(\xvec{h}) = \xvec{\alpha} + \xmat{W}_1 \xvec{\gamma}(\xmat{E} \xvec{h}) - (\xmat{I} - \xmat{W}_2) \xvec{F}(\xvec{h}) = \xvec{0} . \]

Now, we want to know whether there is a solution $\xvec{h}$ of this equation. This is equivalent to the problem of whether or not $\xvec{T}(\xvec{h})$ has a
fixed-point. The existence and the uniqueness of a solution $\xvec{h}$ is an immediate consequence of the Banach fixed-point theorem. Inserting this solution in \eqref{eq:initial} provides a solution $\xvec{Y}$.
\end{proof}

\begin{proof}[Corollary \ref{cor:stationarity2}]
Because $(Y(\xvec{s}_1), \ldots, Y(\xvec{s}_n))^\prime$ is a measurable function of \linebreak $(\varepsilon(\xvec{s}_1), \ldots, \varepsilon(\xvec{s}_n))^\prime$, it is strictly stationary as well.
\end{proof}

\begin{proof}[Theorem \ref{th:moments1}]
\begin{itemize}
\item[a)] Since
\[ \xvec{F} = ( f(\xvec{h}(\xvec{s}_i)) )_{i=1,...,n} = \xvec{\alpha} + \xmat{W}_1 \xvec{\gamma}(\xmat{E} \xvec{h}) + \xmat{W}_2 \xvec{F} \]
its solution $\xvec{h}$ is a function of $\varepsilon(\xvec{s}_1)^2,..., \varepsilon(\xvec{s}_n)^2$, say $h(\xvec{s}_i) = \xi_i(\varepsilon(\xvec{s}_1)^2,..., \varepsilon(\xvec{s}_n)^2)$. Since $Y(\xvec{s}_i) = \varepsilon(\xvec{s}_i) \sqrt{h(\xvec{s}_i)}$ it follows that
\begin{eqnarray*}
-Y(\xvec{s}_i) & = & -\varepsilon(\xvec{s}_i) \sqrt{\xi_i(\varepsilon(\xvec{s}_1)^2,..., \varepsilon(\xvec{s}_n)^2)} \\
& = & -\varepsilon(\xvec{s}_i) \sqrt{\xi_i(\varepsilon(\xvec{s}_1)^2,..., (-\varepsilon(\xvec{s}_i))^2,...,\varepsilon(\xvec{s}_n)^2)} \\
& \stackrel{d}{=} & Y(\xvec{s}_i)
\end{eqnarray*}
since $\xvec{\varepsilon}$ is sign-symmetric. Thus $Y(\xvec{s}_i)$ is a symmetric random variable.

Moreover,
\[ (Y(\xvec{s}_1),\ldots,Y(\xvec{s}_n))^\prime \stackrel{d}{=} (-Y(\xvec{s}_1),\ldots,Y(\xvec{s}_n))^\prime. \]
Thus, $E(Y(\xvec{s}_1)^{2k-1} | Y(\xvec{s}_2),\ldots,Y(\xvec{s}_n)) = E(-Y(\xvec{s}_1)^{2k-1} | Y(\xvec{s}_2),\ldots,Y(\xvec{s}_n))$. Consequently, this quantity is zero.

\item[b)] Now
\begin{eqnarray*}
Y(\xvec{s}_i) Y(\xvec{s}_j) & = & \varepsilon(\xvec{s}_i) \varepsilon(\xvec{s}_j) \sqrt{\xi_i(\varepsilon(\xvec{s}_1)^2,..., \varepsilon(\xvec{s}_n)^2)}
\sqrt{\xi_j(\varepsilon(\xvec{s}_1)^2,..., \varepsilon(\xvec{s}_n)^2)}\\
& \stackrel{d}{=} & - \varepsilon(\xvec{s}_i) \varepsilon(\xvec{s}_j) \sqrt{\xi_i(\varepsilon(\xvec{s}_1)^2,..., \varepsilon(\xvec{s}_n)^2)}
\sqrt{\xi_j(\varepsilon(\xvec{s}_1)^2,..., \varepsilon(\xvec{s}_n)^2)}\\
& = & - Y(\xvec{s}_i) Y(\xvec{s}_j)
\end{eqnarray*}
and thus $Cov(Y(\xvec{s}_i), Y(\xvec{s}_j)) = 0$ for $i \neq j$.

\end{itemize}
\end{proof}

\begin{proof}[Theorem \ref{th:existence}]
The result follows with straightforward calculations.
\end{proof}
\begin{proof}[Corollary \ref{cor:stationarity}]
Follows with the same argumentation as Corollary \ref{cor:stationarity2}.
\end{proof}

\begin{proof}[Theorem \ref{th:moments}]

\begin{itemize}
\item[a)] Because $2|ab| \le a^2 + b^2$, it follows that
\[ E( \left| f^{-1}( c_i + \xvec{d}_i g(\xvec{\varepsilon})) \right|^k \; |\varepsilon(\xvec{s}_i)|^k ) \le
E\left( \left( f^{-1}(c_i + \xvec{d}_i g(\xvec{\varepsilon})) \right)^{2k} \right) + E( \varepsilon(\xvec{s}_i)^{2k} ) . \]
Now, if $f^{-1}$ is convex, it implies that $(f^{-1}(x))^{2k}$ is convex as well. Using the inequality of Jensen, we obtain
\[ E\left( \left( f^{-1}\left(\frac{2c_i}{2} + \frac{2 \xvec{d}_i^\prime g(\xvec{\varepsilon})}{2} \right) \right)^{2k} \right) \le \frac{ \left( f^{-1}(2c_i) \right)^{2k}}{2} + \frac{E\left( \left( f^{-1}(2 \xvec{d}_i^\prime g(\xvec{\varepsilon}) ) \right)^{2k} \right)}{2}. \]
Thus, the result follows.
\item[b)] The first part follows immediately because
\[ Y(\xvec{s}_i) = \sqrt{f^{-1}(c_i + \xvec{d}_i^\prime g(\xvec{\varepsilon}))} \; \varepsilon(\xvec{s}_i) \stackrel{d}{=} -  \sqrt{f^{-1}(c_i + \xvec{d}_i^\prime g(-\xvec{\varepsilon})) }  \; \varepsilon(\xvec{s}_i) = - Y(\xvec{s}_i) . \]
The second part follows as in the proof of Theorem \ref{th:moments1}.
\end{itemize}
\end{proof}

\section*{Partial Derivatives}
\subsection*{SpGARCH}
Using \eqref{epsGARCH}, one can derive $\frac{\partial h(\xvec{s}_i)}{\partial \varepsilon(\xvec{s}_j)}$ for the spGARCH model. To be precise,  we obtain
\[  \frac{\partial h(\xvec{s}_i)}{\partial \varepsilon(\xvec{s}_j)} = \frac{ \partial ( \xmat{I} - \xmat{W}_1 \mbox{diag}(\xvec{\varepsilon}^{(2)}) - \xmat{W}_2 )^{-1} \xvec{\alpha} }{ \partial \varepsilon(\xvec{s}_j) } . \]
From \citet[][8.15]{Harville97}, it follows that
\[ \frac{ \partial ( \xmat{I} - \xmat{W}_1 \mbox{diag}(\xvec{\varepsilon}^{(2)}) - \xmat{W}_2 )^{-1} \xvec{\alpha} }{ \partial \varepsilon(\xvec{s}_j) }
= 2  \varepsilon(\xvec{s}_j) ( \xmat{I} - \xmat{W}_1 \mbox{diag}(\xvec{\varepsilon}^{(2)}) - \xmat{W}_2 )^{-1}  \;\]
\[ \hspace*{4cm} \cdot \;  (\xvec{0},\ldots,\xvec{0}, \xvec{w}_{1,j}, \xvec{0},\ldots,\xvec{0} ) \;
( \xmat{I} - \xmat{W}_1 \mbox{diag}(\xvec{\varepsilon}^{(2)}) - \xmat{W}_2)^{-1} \xvec{\alpha},   \]
where $\xvec{w}_{1,j}$ denotes the $j$-th column of $\xmat{W}_1$.

\subsection*{Hybrid Spatial GARCH}
Similarly, we can derive the derivative for the H-spGARCH model using \eqref{epsHARCH}, that is,
\[ \frac{\partial h(\xvec{s}_i)}{\partial \varepsilon(\xvec{s}_j)} = \frac{\partial \log(h(\xvec{s}_i))}{\partial \varepsilon(\xvec{s}_j)} h(\xvec{s}_i), \]
where $\frac{\partial \log(h(\xvec{s}_i))}{\partial \varepsilon(\xvec{s}_j)} = 2 d_{ij}/\varepsilon(\xvec{s}_j)$ and $(\xmat{I} - \xmat{W}_1 - \xmat{W}_2)^{-1} \xmat{W}_1 =
(d_{ij})$.

\subsection*{Exponential and Logarithmic Spatial GARCH}
Finally, for the spatial E-spGARCH, it holds using \eqref{spEGARCH3} that
\[ \xvec{Z} = (\xmat{I} - \xmat{W}_2)^{-1} (\xmat{W}_1 \xvec{G} + \xvec{\alpha}), \]
where $\xvec{G} = ( g(\varepsilon(\xvec{s}_1),\ldots, g(\varepsilon(\xvec{s}_n)))^\prime$ and
\[ \frac{\partial h(\xvec{s}_i)}{\partial \varepsilon(\xvec{s}_j)} = \frac{\partial \log(h(\xvec{s}_i))}{\partial \varepsilon(\xvec{s}_j)} h(\xvec{s}_i) = c_{ij} g^\prime(\varepsilon(\xvec{s}_j)) h(\xvec{s}_i), \]
where $(\xmat{I} - \xmat{W}_2)^{-1} \xmat{W}_1 = (c_{ij})$.

% \section*{References}

\end{document}